\documentclass[11pt,reqno]{amsart}
\usepackage{amssymb}
\usepackage{amsfonts}
\usepackage{amsmath}
\usepackage{stmaryrd}
\usepackage{physics}
\usepackage{braket}
\usepackage{dsfont}
\usepackage{bbold}
\usepackage{graphicx}
\usepackage{svg}
\usepackage{relsize}
\usepackage{makecell}
\usepackage[table,xcdraw]{xcolor}
\usepackage{enumerate}
\usepackage[pagebackref, colorlinks = true, linkcolor = blue, urlcolor  = blue, citecolor = red]{hyperref}
\usepackage[margin=1in]{geometry}
\usepackage{enumitem}
\usepackage{mathtools}
\usepackage{graphbox}
\usepackage{comment}
\usepackage{float}
\usepackage{hhline}
\usepackage[capitalize]{cleveref}
\usepackage{accents}

\newtheorem{theorem}{Theorem}[section]
\newtheorem{proposition}[theorem]{Proposition}
\newtheorem{corollary}[theorem]{Corollary}
\newtheorem{lemma}[theorem]{Lemma}

\newtheorem*{conjecture*}{Conjecture}

\newcommand\vertarrowbox[3][6ex]{
  \begin{array}[t]{@{}c@{}} #2 \\
  \left\uparrow\vcenter{\hrule height #1}\right.\kern-\nulldelimiterspace\\
  \makebox[0pt]{\scriptsize#3}
  \end{array}
}

\theoremstyle{definition}

\theoremstyle{definition}

\theoremstyle{definition}

\theoremstyle{definition}

\definecolor{darkgreen}{rgb}{0,0.392,0}

\newtheorem{innercustomgeneric}{\customgenericname}
\providecommand{\customgenericname}{}
\newcommand{\newcustomtheorem}[2]{
  \newenvironment{#1}[1]
  {
   \ifdefined\crefalias\crefalias{innercustomgeneric}{#2}\fi
   \renewcommand\customgenericname{#2}
   \renewcommand\theinnercustomgeneric{##1}
   \innercustomgeneric
  }
  {\endinnercustomgeneric}
  \ifdefined\crefname\crefname{#2}{#2}{#2s}\fi
}

\newcustomtheorem{customthm}{Theorem}
\newcustomtheorem{customlemma}{Lemma}

\newcommand{\A}{\mathcal{A}}

\newcommand{\Hi}{\mathcal{H}}

\newcommand{\K}{\mathcal{K}}

\newcommand{\Comp}{\mathbb{C}}
\newcommand{\re}{{\rm Re\,}}

\newcommand{\diag}{\mathrm{diag}}
\newcommand{\id}{\mathrm{id}}
\newcommand{\la}{\langle}
\newcommand{\ra}{\rangle}
\newcommand{\Om}{\Omega}
\newcommand{\om}{\omega}

\newcommand{\vect}{{\rm vec\,}}

\newcommand{\Fop}{\mathcal Q}

\title[2-distillability of Werner states]{A solution to 2-copy distillability of Werner states}
\author{Jinshi Fu}
\email{jinshi.fu@whu.edu.cn}
\author{Li Gao}
\email{gao.li@whu.edu.cn}
\author{Sang-Jun Park}
\email{sjpark@whu.edu.cn}
\address{School of Mathematics and Statistics, Wuhan University, 
430070, Wuhan, Hubei, China}

\begin{document}

\begin{abstract}
Entanglement distillation is a fundamental task in quantum information theory. In this work, we prove that Werner states in arbitrary dimension are $2$-copy distillable if and only if they are $1$-copy distillable. This answers the longstanding open question of the $2$-copy distillability of Werner states. This is an important step on determining whether every non-positive partial transpose (NPT) state is distillable, which remains one of the central open problems in the field of entanglement distillation. 
\end{abstract}

\maketitle

\section{Introduction}

Entanglement is a fundamental resource in quantum information theory and underlies many quantum information-processing tasks \cite{horodecki2009quantum}. In realistic settings, however, entanglement is inevitably degraded by noise and interactions with the environment. This makes it necessary to extract high-quality entanglement from imperfect quantum states and motivates the theory of \emph{entanglement distillation}: two distant parties use \emph{local operations and classical communication (LOCC)} to convert several copies of a mixed state into a smaller number of nearly maximally entangled pairs \cite{bennett1996purification,bennett1996mixed}.

A central problem in entanglement theory is to determine which bipartite states are distillable. The operational notion admits the following mathematical criterion \cite{HHH97,HHH98,DCLB00}: a bipartite state $\rho_{AB}$ is \emph{$r$-copy distillable} if and only if there exists a vector $\ket{\psi}_{A^rB^r}$ of Schmidt rank at most two such that
    $$\bra{\psi}(\rho_{AB}^{\Gamma})^{\otimes r}\ket{\psi}<0,$$
where $\rho_{AB}^{\Gamma}:=(\id_A\otimes \top_B)(\rho_{AB})$ denotes the partial transpose. A state is called \emph{distillable} if it is $r$-copy distillable for some finite $r$. Every state with positive partial transpose (PPT) is undistillable \cite{HHH98}, and hence every distillable state must have non-positive partial transpose (NPT). The converse question, whether every NPT state is distillable, is one of the longstanding open problems in quantum information theory. Following the discovery of PPT \emph{bound entangled} states \cite{Wor76,Cho82,Hor97,HHH98} (i.e., entangled but undistillable states), the works \cite{DCLB00,DVSS+00} raised the possibility that bound entanglement may also occur in the NPT regime. Despite more than two decades of investigation, this \emph{NPT distillation problem} remains open; see \cite{PPHH10,HRZ22,HPS25} and the references therein for the detailed history on the problem.

Among bipartite quantum states, the \emph{Werner states} \cite{Werner1989}
$$
    \rho_{\alpha}=\rho_{\alpha}^{(d)}
    :=\frac{I_{d^2}+\alpha F_d}{d^2+\alpha d},
    \qquad \alpha\in[-1,1],
$$
where $F_d$ is the flip operator on $\mathbb C^d\otimes\mathbb C^d$, play a distinguished role. By the reduction via the associated twirling arguments \cite{HH99,DCLB00,DVSS+00}, the general NPT distillation problem can be reduced to this one-parameter family: every NPT state is distillable if and only if every entangled Werner state is distillable. Understanding the dependence of the distillability of $\rho_\alpha$ on $\alpha$ is therefore a canonical test of the general NPT distillation problem.

The basic thresholds are known. First, the following three conditions are equivalent \cite{Werner1989,VW02}: $\rho_\alpha$ is PPT, $\rho_\alpha$ is separable, and $\alpha\in[-1/d,1]$. Moreover, $\rho_\alpha$ is 1-copy distillable if and only if $\alpha\in[-1,-1/2)$ \cite{Tom85,DCLB00}. Thus, for $d\geq3$, the unresolved Werner regime is the NPT but 1-copy-undistillable interval
$$
    -\frac12\leq \alpha< -\frac1d.
$$
The 2-copy question is the first genuinely collective case and has proved to be highly nontrivial. It has long been conjectured that the 2-copy threshold coincides with the 1-copy threshold. For $d=2$, the relevant interval is empty because $\rho_\alpha$ is already PPT for $\alpha\geq-1/2$. For $d=3$, robust semidefinite-programming and entanglement-witness methods give strong numerical evidence for the conjecture \cite{VD06}, with further algebraic reformulations and evidence in \cite{Djok16}. The case $d=4$ has led to several equivalent matrix-inequality formulations \cite{PPHH10} and remains an important test case; see \cite{PPHH10,QCCS21,LC26} for partial results. In arbitrary dimension, the strongest previously known exact results \cite{DCLB00,RW25} assert that $\rho_\alpha^{(d)}$ is 2-copy undistillable whenever 
    $$\alpha\geq-\max\left\{\frac14+\frac{1}{2d},\frac13\right\}.$$
We refer to \cite{EVWW01,JK11,CD11,PBHS13,Djok16,MHRW16,QGWG24,Ric25} for several other approaches to the NPT distillation problem.

{
Our work was inspired in particular by the partial-trace approach of Rico and Wolf. Rico reformulated finite-copy undistillability of Werner states in terms of partial-trace inequalities and recovered two-copy undistillability for $\alpha\geq-1/4$ \cite{Ric25}. Rico and Wolf subsequently improved this dimension-independent bound to $\alpha\geq-1/3$ \cite{RW25}.}

\medskip

In this paper, we settle the $2$-copy distillability problem for Werner states in every dimension $d\ge 2$. We prove that every $1$-copy-undistillable Werner state is also $2$-copy-undistillable, thereby obtaining the complete 2-copy threshold and resolving the longstanding question posed in \cite{DCLB00,PPHH10} and \cite[Problem~5]{HRZ22}.

\begin{theorem}\label{thm-main}
For every $d\geq2$, the Werner state $\rho_\alpha$ is $2$-copy undistillable if and only if
    $$-\frac12\leq\alpha\leq1.$$
In particular, it is NPT and $2$-copy undistillable whenever $\displaystyle -\frac12\leq\alpha< -\frac1d.$
\end{theorem}

Theorem~\ref{thm-main} shows that, within the Werner family, access to a second copy never unlocks distillability: the 1-copy and 2-copy thresholds coincide exactly. Consequently, if a Werner state in the interval $-\frac12\leq\alpha<-\frac1d$ is distillable at all, then at least three copies are required. Thus our theorem resolves the first nontrivial finite-copy level of the canonical NPT distillation problem, while leaving the higher-copy and asymptotic questions open.

Our proof is based on a sharp geometric estimate for symmetric and antisymmetric tensor subspaces. Let
$$
    Q:=\Pi_{\mathcal A}^{(A_1B_1)}\otimes\Pi_{\mathcal A}^{(A_2B_2)}
$$
be the projection onto the tensor product of the two antisymmetric subspaces. We first show the dimension-independent optimization (Theorem \ref{thm-AntisymmProj-S2norm})
$$
    \sup_{\substack{\|\psi\|=1\\
    \operatorname{SR}_{A_1A_2:B_1B_2}(\psi)\leq2}}
    \langle\psi|Q|\psi\rangle=\frac12
$$
by applying the ideas in \cite{PPHH10,JK10}.
This use of symmetric-subspace structure is closely related to the general role of permutation symmetry in quantum information theory \cite{harrow2013church}.

To prove Theorem~\ref{thm-main}, the only new undistillability statement that must be established is the \emph{boundary case} $\alpha=-\frac12$. The proof consists of two additional steps. First, after vectorizing a Schmidt-rank-two test vector
    $$|\psi\ra = |\vect V_1 \ra_{A_1A_2}|\vect W_1\ra_{B_1B_2} + |\vect V_2\ra_{A_1A_2}|\vect W_2\ra_{B_1B_2},$$
the undistillability condition at $\alpha=-1/2$ becomes a sharp block-operator inequality associated with a Hilbert--Schmidt orthonormal pair of matrices $V=(V_1,V_2)$ (Lemma \ref{lem:endpoint-reduction}). Then we relate each block in the operator inequality involving the above projection $Q$ in geometric perspective. Specifically, it turns out that the zeroth, first, and second-order variations of the Hilbert-Schmidt norm of the restricted projection
    $$Q\big|_{\Hi^{(A_1A_2)}\vee \Hi^{(B_1B_2)}}, \qquad \Hi={\rm span}\{|\vect V_1\ra, |\vect V_2\ra\}$$
provide the surprising connection between these blocks (Lemmas \ref{lem:projection-kernel}, \ref{lem:first-variation}, and \ref{cor-DQV-K}), where the advertised sharp projection estimate allows us to conclude the proof.

\medskip

The remainder of the paper is organized as follows. Section~\ref{sec-symm-antisymm} recalls the necessary preliminaries on Schmidt decompositions and symmetric and antisymmetric subspaces, and proves the sharp projection estimate in Lemma~\ref{lem-AntisymmProj} and Theorem~\ref{thm-AntisymmProj-S2norm}. Section~\ref{sec-distillability} reformulates the endpoint problem as a block-operator inequality and completes the proof through the corresponding reduction identity of zeroth-, first-, and second-level variation. The final section briefly discusses the remaining higher-copy problem.

\section{Symmetric and antisymmetric subspaces} \label{sec-symm-antisymm}

Let $\mathcal{H}_A$ and $\mathcal{H}_B$ be finite-dimensional Hilbert spaces and let $n:=\min(\dim \mathcal{H}_A,\dim \mathcal{H}_B)$. For any bipartite pure state $|\psi\rangle\in \mathcal{H}_A\otimes \mathcal{H}_B$ (i.e., a vector $\psi$ with $\|\psi\|=1$), it admits a \emph{Schmidt decomposition} \cite{nielsen2010quantum}
$$
|\psi\rangle=\sum_{i=1}^{n}s_i|v_i\rangle\otimes|w_i\rangle,
$$
where ${|v_i\rangle}$ and ${|w_i\rangle}$ are orthonormal families in $\mathcal{H}_A$ and $\mathcal{H}_B$, respectively, $s_1\geq s_2\geq \cdots \geq s_n\geq 0$, and $\sum_i s_i^2=1$. The numbers $s_1,\ldots, s_n$ are called the \emph{Schmidt coefficients}, and the largest number $r$ such that $s_r>0$ is called the \emph{Schmidt rank} of $|\psi\rangle$. In particular, a pure state is \emph{separable} if and only if its Schmidt rank is one. More generally, a mixed state $\rho$ is separable if it can be written as a convex combination of product states,
$$
\rho=\sum_j p_j\rho_j^A\otimes\rho_j^B,
\qquad
p_j\geq 0,\quad \sum_jp_j=1.
$$

Let $\Hi$ be a finite-dimensional Hilbert space, and let us denote by
$F_{\Hi}:\Hi\otimes \Hi\to \Hi\otimes \Hi$ the \emph{flip operator}, defined on elementary
tensors by
$$
    F_{\Hi}(\ket{\xi}\otimes \ket{\eta})=\ket{\eta}\otimes \ket{\xi},
    \qquad \ket{\xi},\ket{\eta}\in \Hi.
$$
$F_{\Hi}$ is a Hermitian unitary, hence it has two eigenvalues, $1$ and $-1$. Its spectral projections are called symmetric and antisymmetric projections on $\Hi \otimes \Hi$, defined as
$$
    \Pi_{\mathcal S}:=\frac{I_{\Hi\otimes \Hi}+F_{\Hi}}{2},
    \qquad
    \Pi_{\mathcal A}:=\frac{I_{\Hi\otimes \Hi}-F_{\Hi}}{2},
$$
whose ranges are the eigenspaces
$$
    \Hi \vee \Hi:=\{ \ket{\psi}\in \Hi \otimes \Hi \ |\ F_{\Hi} \ket{\psi}=\ket{\psi} \},
    \qquad
     \Hi \wedge \Hi:=\{ \ket{\psi}\in \Hi\otimes \Hi \ |\ F_{\Hi} \ket{\psi}=-\ket{\psi} \},
$$
called \emph{symmetric subspace} and \emph{anti-symmetric subspace} of $H\otimes H$, respectively. Given an orthonormal basis
$\{v_1,\ldots, v_n\}$ of $\Hi \cong \mathbb{C}^n$, the flip operator can be explicitly written as
$$F_{\Hi}=\sum_{i,j=1}^n \ket{v_i}\bra{v_j}\otimes \ket{v_j}\bra{v_i}.$$
Note that the above expression is independent of the choice of the orthonormal basis.
Furthermore, the two sets
\begin{align*}
    &\{\ket{v_i}\otimes \ket{v_i}:1\leq i\leq n\} \cup
    \left\{\frac{\ket{v_i} \otimes \ket{v_j}+\ket{v_j} \otimes \ket{v_i}}{\sqrt 2}:1\leq i<j\leq n\right\}, \\
    &\left\{\frac{\ket{v_i} \otimes \ket{v_j}-\ket{v_j} \otimes \ket{v_i}}{\sqrt 2}:1\leq i<j\leq n\right\}
\end{align*}
are orthonormal bases of $\Hi \vee \Hi$ and $\Hi \wedge \Hi$, respectively. 

For simplicity, we write $F_n:=F_{\Comp^n} = \sum_{i,j=1}^n |i\ra\la j|\otimes | j\ra \la i| \in M_n(\Comp)^{\otimes 2}$. 
Furthermore, we identify the space $\Comp^n\otimes \Comp^n$ with $M_n(\Comp)$ through the convention
$$
    \ket{\vect X}:=\sum_{i,j=1}^n X_{ij}\ket{ij}=(X\otimes I_n)\ket{\Omega_n}, \quad X=(X_{ij})_{1\leq i,j\leq n}\in M_n(\Comp),
$$
where $\ket{\Omega_n}:=\sum_{i=1}^n \ket{i}\otimes\ket{i}=\ket{\vect I}$ denotes the unnormalized \emph{maximally entangled state}. For rank-one matrices, we have $|\vect (|\psi\ra\la \varphi|)\ra:= |\psi\ra\otimes |\overline{\varphi}\ra$ where $|\psi\ra\la \varphi|:= \psi\varphi^*$ and $\overline{\varphi}$ denotes the entrywise complex conjugate vector of $\varphi$, and extended linearly.

\medskip

Let us first recall the following property of the Schmidt decomposition of symmetric vectors.

\begin{proposition} \label{prop-Schmidt-Symm}
For any symmetric vector $\ket{\psi}\in \Hi\vee \Hi$, there are
nonnegative numbers $s_1\geq s_2\geq \cdots \geq s_n\geq 0$ and an
orthonormal basis $\{\ket{v_1},\ldots,\ket{v_n}\}$ of $\Hi$ such that
\begin{equation} \label{eq-Schmidt-Symm}
    \ket{\psi} = \sum_{i=1}^n s_i \ket{v_i}\otimes \ket{v_i}.
\end{equation}
\end{proposition}
\begin{proof}
We may identify $\Hi\cong \Comp^n$ and write $\ket{\psi}=\ket{\vect X}$ for some matrix $X\in M_n(\Comp)$. Then the condition
$F_n\ket{\psi}=\ket{\psi}$ is exactly the condition $X^{\top}=X$. By
Takagi's factorization for complex symmetric matrices (see, e.g., {\cite[Theorem 4.4.16]{Horn2012matrix}}), there exist a unitary
matrix $U=(v_1\ \cdots\ v_n)$ and a diagonal matrix
$D=\diag(s_1,\ldots,s_n)$ with $s_i\geq 0$ such that
    $$X=UDU^{\top}=\sum_{i=1}^n s_i v_i v_i^{\top} = \sum_{i=1}^n s_i |v_i\ra \la \overline{v_i}|.$$
Vectorizing this identity gives
\eqref{eq-Schmidt-Symm}. Reordering the columns of $U$ if necessary gives
$s_1\geq \cdots\geq s_n$.
\end{proof}

We shall also use the same notation for Hilbert \emph{subspaces}. Let us first label the Hilbert spaces
    $$\Hi_{A_1}=\Hi_{A_2}=\Hi_{B_1}=\Hi_{B_2}=\Comp^d, \qquad \Hi_{A_1A_2} \otimes \Hi_{B_1B_2} = (\Comp^d)^{\otimes 2}_{12} \otimes (\Comp^d)^{\otimes 2}_{34}.$$
We are interested in the case
$\Hi \subseteq \Comp^d\otimes \Comp^d\cong \Comp^{d^2}$, and then $\Hi ^{(12)}\vee \Hi^{(34)}$ denotes the
symmetric subspace of $\Hi^{(12)}\otimes \Hi^{(34)}\subseteq (\Comp^d\otimes \Comp^d) \otimes (\Comp^d\otimes \Comp^d)$ with respect to the flip operator
$$
    F_{d^2}^{(12:34)}=F^{(13)}F^{(24)}\in M_d(\Comp)^{\otimes 4},
$$
between the two copies $12\leftrightarrow 34$. 
Here $F^{(13)}=F_d^{(13)}\otimes I_{d^2}^{(24)}$ is the flip operator between the first and third $\Comp^d$
factors, and $F^{(24)}= I_{d^2}^{(13)} \otimes F_d^{(24)}$ is the flip operator between the second and fourth
factors. In this case, the flip operator $F_{\Hi}$ on $\Hi^{(12)}\otimes \Hi^{(34)}$ simply corresponds to the restriction $F_{\Hi}=F_{d^2}^{(12:34)}\big|_{\Hi \otimes \Hi}$. 

\medskip

Let us begin with a lemma which would be useful throughout this paper.

\begin{lemma}\label{lemma:4trace}
For all matrices $X_1, X_2, Y_1, Y_2\in M_d(\Comp)$, one has
\begin{align}
    (\bra{\vect X_1}_{12}\otimes \bra{\vect X_2}_{34})F^{(13)}
    (\ket{\vect Y_1}_{12}\otimes \ket{\vect Y_2}_{34})=\Tr(X_1^*Y_2X_2^*Y_1), \label{eq-4trace1}\\
    (\bra{\vect X_1}_{12}\otimes \bra{\vect X_2}_{34})F^{(24)}
    (\ket{\vect Y_1}_{12}\otimes \ket{\vect Y_2}_{34})=\Tr(X_1^*Y_1X_2^*Y_2), \label{eq-4trace2}
\end{align}
\end{lemma}
\begin{proof}
Recall that $|\vect X_i\ra = (X_i\otimes I_d)|\Om_d\ra$ and $|\vect Y_i\ra = (Y_i\otimes I_d)|\Om_d\ra$ ($i=1,2$), and we can naturally identify
    $$|\Om_d^{(12)}\ra \otimes |\Om_d^{(34)}\ra =  \sum_{i,j=1}^d |ii\ra_{12}\otimes |jj\ra_{34} \cong  \sum_{i,j=1}^{d}|ij\ra_{13}\otimes |ij\ra_{24} = |\Om_{d^2}^{(13:24)}\ra .$$
Then one has
\begin{align*}
    &(\bra{\vect X_1}_{12}\otimes \bra{\vect X_2}_{34})F^{(13)}
    (\ket{\vect Y_1}_{12}\otimes \ket{\vect Y_2}_{34})\\
    &= \big\la \Om_d^{(12)}\otimes \Om_d^{(34)} \big| (X_1^{*(1)}\otimes X_2^{*(3)}\otimes I_{d^2}^{(24)})(F_d^{(13)}\otimes I_{d^2}^{(24)}) (Y_1^{(1)}\otimes Y_2^{(3)}\otimes I_{d^2}^{(24)}) \big| \Om_d^{(12)}\otimes \Om_d^{(34)} \big\ra\\
    &= \big \la \Om_{d^2}^{(13:24)} \big|(X_1^{*(1)}\otimes X_2^{*(3)})F_d^{(13)}(Y_1^{(1)}\otimes Y_2^{(3)}) \otimes I_{d^2}^{(24)} \big|\Om_{d^2}^{(13:24)} \big\ra\\
    &= \Tr \big((X_1^*\otimes X_2^*)F_d(Y_1\otimes Y_2)\big)= \Tr \big((Y_1X_1^*\otimes Y_2X_2^*)F_d\big)\\
    &= \Tr(Y_1X_1^*Y_2X_2^*) = \Tr(X_1^*Y_2X_2^*Y_1)
\end{align*}
which shows Eq. \eqref{eq-4trace1}, where we further applied the two identities
    $$\la \Om_n|Z\otimes I_n|\Om_n\ra = \Tr(Z), \quad \Tr((P\otimes Q)F_d)=\Tr(PQ)$$
for $Z\in M_n(\Comp)$ and $P,Q\in M_d(\Comp)$.
The proof of Eq. \eqref{eq-4trace2} is analogous. We also refer to Figure \ref{fig-VecId} for a graphical proof of Eq. \eqref{eq-4trace1} using the tensor diagram notation from \cite{collins2010random,collins2016random}.
\end{proof}

\begin{figure}[htbp]
    \centering

    \includegraphics[width=1\linewidth]{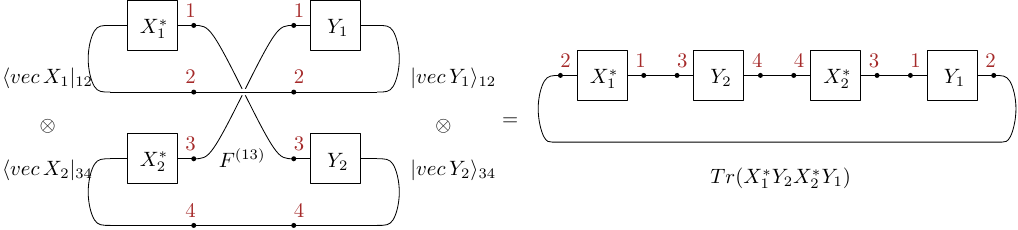} 
    
    \caption{A graphical proof of Eq. \eqref{eq-4trace1}}
    \label{fig-VecId}
\end{figure}

\begin{lemma} \label{lem-AntisymmProj}
Let $\Hi \subseteq \Comp^d\otimes \Comp^d$ be a subspace with $\dim(\Hi)\le 2$.
Then, for every symmetric vector
$\ket{\psi}\in \Hi^{(12)}\vee \Hi^{(34)}\subseteq (\Comp^d)^{\otimes 4}$, one has
\begin{equation} \label{eq-AntisymmIneq1}
    \| (\Pi_{\A}^{(13)}\otimes \Pi_{\A}^{(24)})\ket{\psi}\|^2
    \leq \frac{1}{2}\|\psi\|^2.
\end{equation}
\end{lemma}
\begin{proof}
Since $\ket{\psi}$ is
symmetric under the exchange of the two copies $12\leftrightarrow 34$, we have
\begin{equation} \label{eq-symm-identity}
    F^{(12:34)}\ket{\psi}=F^{(13)}F^{(24)}\ket{\psi}=\ket{\psi},
    \qquad F^{(13)}\ket{\psi}=F^{(24)}\ket{\psi}.
\end{equation}
Using $\Pi_{\A}=(I_{d^2}-F_d)/2$, we obtain
\begin{align*}
    \|(\Pi_{\A}^{(13)}\otimes \Pi_{\A}^{(24)})\ket{\psi}\|^2
    &=\bra{\psi} \Pi_{\A}^{(13)}\otimes \Pi_{\A}^{(24)}\ket{\psi} \\
    &=\frac{1}{4}\Big(\braket{\psi|\psi}-\bra{\psi}F^{(13)}\ket{\psi}
        -\bra{\psi}F^{(24)}\ket{\psi}
        +\bra{\psi}F^{(13)}F^{(24)}\ket{\psi}\Big) \\
    &=\frac{1}{2}\Big(\|\psi\|^2-\bra{\psi}F^{(24)}\ket{\psi}\Big).
\end{align*}
Therefore, it suffices to prove that $\bra{\psi}F^{(24)}\ket{\psi}\geq 0$. Indeed, by Proposition \ref{prop-Schmidt-Symm} applied inside the Hilbert space $\Hi$,
one can write
$$
    \ket{\psi}=s_1\ket{v_1}_{12}\otimes\ket{v_1}_{34}
              +s_2\ket{v_2}_{12}\otimes\ket{v_2}_{34},
$$
where $s_1,s_2\geq 0$ and $\ket{v_1},\ket{v_2}$ are orthonormal vectors in
$\mathcal{H}\subseteq \Comp^d\otimes \Comp^d$; if $\dim \Hi=1$, we take $s_2=0$. Write
$\ket{v_i}=\ket{\vect V_i}$ for the matrix $V_i\in M_d(\Comp)$, and 
$$ \langle\psi|F^{(24)}|\psi\rangle
=s_1^2\operatorname{Tr}\bigl((V_1^*V_1)^2\bigr)\geq0.$$
For $\dim \Hi=2$, by Lemma \ref{lemma:4trace}
\begin{align*}
    \bra{\psi}F^{(24)}\ket{\psi}
    &=s_1^2\Tr\big((V_1^*V_1)^2\big)
      +2s_1s_2\re\Tr(V_1^*V_2V_1^*V_2)
      +s_2^2\Tr\big((V_2^*V_2)^2\big) \\
    &\geq s_1^2\|V_1\|_4^4
       -2s_1s_2\|V_1\|_4^2\|V_2\|_4^2
       +s_2^2\|V_2\|_4^4 \\
    &=\big(s_1\|V_1\|_4^2-s_2\|V_2\|_4^2\big)^2\geq 0,
\end{align*}
where we used H\"older's inequality for Schatten norms,
$$
    |\Tr(PQRS)|\leq \|P\|_4\|Q\|_4\|R\|_4\|S\|_4.
$$
Substituting this nonnegativity into the previous expression for
$\| (\Pi_{\A}^{(13)}\otimes \Pi_{\A}^{(24)}) \ket{\psi}\|^2$ gives \eqref{eq-AntisymmIneq1}.
\end{proof}

We now present our key symmetric-antisymmetric estimate.

\begin{theorem} \label{thm-AntisymmProj-S2norm}
Let $d\ge 2$. For every bipartite pure state $|\psi\ra \in (\Comp^d)^{\otimes 2}_{12}\otimes (\Comp^d)^{\otimes 2}_{34}$ with ${\rm SR}(\psi)\leq 2$, one has
\begin{equation} \label{eq-AntisymmIneq}
    \la \psi| (\Pi_{\A}^{(13)}\otimes \Pi_{\A}^{(24)})\ket{\psi}
    \leq \frac{1}{2}.
\end{equation}
In other words, $\|\Pi_{\A}^{(13)}\otimes \Pi_{\A}^{(24)}\|_{S(2),12:34}=\frac{1}{2}$ with the $S(2)$-norm introduced in \cite{JK10}.
\end{theorem}
\begin{proof}We utilize the idea of \cite{PPHH10, JK10}.
Set $Q:=\Pi_\A^{(13)}\otimes\Pi_\A^{(24)}$ for simplicity.
Since $Q$ is an orthogonal projection, for every unit vector
$|\psi\rangle$ one has
    $$\langle\psi|Q|\psi\rangle =\|Q\psi\|^2  =\max_{\substack{\phi\in\operatorname{Ran}(Q)\\ \|\phi\|=1}} |\langle\psi|\phi\rangle|^2.$$
 Consequently, one can write
\begin{align*}
    \max_{\substack{\|\psi\|=1\\
    \operatorname{SR}_{12:34}(\psi)\leq 2}}
     \langle\psi|Q|\psi\rangle
    =
    \max_{\substack{\|\psi\|=1\\
    \operatorname{SR}_{12:34}(\psi)\leq 2}}
    \max_{\substack{\phi\in\operatorname{Ran}(Q)\\ \|\phi\|=1}}
     |\langle\psi|\phi\rangle|^2 
    =
    \max_{\substack{\phi\in\operatorname{Ran}(Q)\\ \|\phi\|=1}}
    \max_{\substack{\|\psi\|=1\\
    \operatorname{SR}_{12:34}(\psi)\leq 2}}
     |\langle\psi|\phi\rangle|^2.
\end{align*}
Fix a unit vector $\phi\in\operatorname{Ran}(Q) = (\Comp^d\wedge \Comp^d)_{13}\otimes (\Comp^d\wedge \Comp^d)_{24}$. Then $F^{(13)}\phi= F^{(24)}\phi=-\phi$, and hence
    $$F^{(12:34)}\phi=F^{(13)}F^{(24)}\phi=\phi.$$
Thus $\phi$ is symmetric across the partition $12:34$. By
Proposition \ref{prop-Schmidt-Symm}, it has a Schmidt decomposition of the form
$$
|\phi\rangle=\sum_{i=1}^{d^2}s_i
 |v_i\rangle_{12}|v_i\rangle_{34},
\qquad
s_1\geq s_2\geq\cdots\geq 0,
\qquad
\sum_i s_i^2=1.
$$
Then we have, by \cite[Theorem 3.3]{JK10},
    $$\max_{\substack{\|\eta\|=1\\ \operatorname{SR}_{12:34}(\eta)\leq 2}}|\langle\eta|\phi\rangle|^2=s_1^2+s_2^2,$$
and the equality above is attained by the optimal state
    $$|\psi_\phi\rangle := \frac{s_1|v_1\rangle_{12}|v_1\rangle_{34} +s_2|v_2\rangle_{12}|v_2\rangle_{34}} {\sqrt{s_1^2+s_2^2}}.$$
Now by taking the $2$-dimensional subspace $\Hi:=\operatorname{span}\{v_1,v_2\}$, we have $\psi_\phi\in\mathcal H^{(12)}\vee\mathcal H^{(34)}$ and $\operatorname{SR}_{12:34}(\psi_\phi)\leq 2$.
Therefore, by Lemma \ref{lem-AntisymmProj} and the condition $Q\phi=\phi$, we obtain
$$
\begin{aligned}
\max_{\substack{\|\eta\|=1\\
\operatorname{SR}_{12:34}(\eta)\leq 2}}
 |\langle\eta|\phi\rangle|^2
=|\langle\psi_\phi|\phi\rangle|^2                              
=|\langle Q\psi_\phi|\phi\rangle|^2                            
\leq \|Q\psi_\phi\|^2\|\phi\|^2                                
\leq \frac12.
\end{aligned}
$$
Taking the maximum over all unit vectors
$\phi\in\operatorname{Ran}(Q)$ proves
$
\langle\psi|Q|\psi\rangle\leq\frac12
$
for every unit vector $\psi$ with
$\operatorname{SR}_{12:34}(\psi)\leq 2$. On the other hand, the bound in Eq. \eqref{eq-AntisymmIneq} is also attained. Indeed, for a vector
    $$ |\psi_0\rangle = \frac{1}{\sqrt2}\bigl(|11\rangle_{12}|22\rangle_{34} + |22\rangle_{12}|11\rangle_{34} \bigr) \in (\Comp^d)^{\otimes 4},$$
one has $\operatorname{SR}_{12:34}(\psi_0)=2$ and $Q|\psi_0\rangle = \frac{1}{\sqrt2}|a\rangle_{13}|a\rangle_{24}$ where $|a\ra = \frac{1}{\sqrt 2}(|12\rangle-|21\rangle)$. Therefore,
$$
\langle\psi_0|Q|\psi_0\rangle
=\|Q\psi_0\|^2
=\frac12,
$$
which shows the assertion $\displaystyle \|Q\|_{S(2)} = \max_{\substack{\|\psi\|=1\\\operatorname{SR}_{12:34}(\psi)\leq 2}}\langle\psi|Q|\psi\rangle = \frac12$.
\end{proof}

\section{Proof of the main theorem}\label{sec-distillability}
First of all, it suffices to prove the 2-copy undistillability of Werner states
    $$ \rho_{\alpha}=\frac{I_{d^2}+\alpha F_d}{d^2+\alpha d}$$
at the \emph{endpoint} $\alpha=-\frac12$ (see e.g., \cite[Lemma 4]{DVSS+00}). Recall that  $\rho_{\alpha}$ is $2$-copy undistillable if and only if $\la \psi|(\rho_{\alpha}^{\Gamma})^{\otimes 2}|\psi\ra \geq 0$ for all bipartite vectors $\psi\in (\Comp^d)^{\otimes 2}_{12} \otimes (\Comp^d)^{\otimes 2}_{34}$ with ${\rm SR}(\psi)\leq 2$. By the Schmidt decomposition, we may write
\begin{equation} \label{eq-bipartite-Schmidt}
    |\psi\ra = |v_1\ra_{12} \otimes |w_1\ra_{34} + |v_2\ra_{12} \otimes |w_2\ra_{34}
\end{equation}
where $v_1,v_2\in (\Comp^d)^{\otimes 2}$ are orthonormal while $w_1,w_2\in (\Comp^d)^{\otimes 2}$ are \emph{arbitrary}. We aim to show the $2$-copy undistillability of $\rho_{-1/2}$ by showing the inequality $\la \psi|(\rho_{-1/2}^{\Gamma})^{\otimes 2}|\psi\ra \geq 0$ for any choice of orthonormal pairs $(v_1,v_2)$.
The argument has an interesting geometric structure.  The Schmidt-rank-two test vector gives an operator inequality on a pair of
matrices $V=(V_1,V_2)$.  Relative to the decomposition into the radial
direction $\mathbb C V$ and its orthogonal complement, this operator
inequality is equivalent to a Schur-complement condition.  We then relate
the diagonal, off-diagonal, and transverse terms in that Schur complement
by the value, first derivative, and second variation of the Hilbert--Schmidt norm with the antisymmetric projection studied in Section \ref{sec-symm-antisymm},
 which combined with Theorem \ref{thm-AntisymmProj-S2norm} completes the endpoint argument.  

\subsection{The endpoint as a block-operator inequality}
We first provide several reductions of $2$-copy undistillability of the endpoint Werner state $\rho_{-1/2}$ in terms of operator inequalities. We shall consider $\mathcal K:=M_d(\mathbb C)\oplus M_d(\mathbb C)$ as a Hilbert space with inner product
    $$
     \langle V,W\rangle_{\mathcal K}
     :=\sum_{i=1}^2\operatorname{Tr}(V_i^*W_i),
     \qquad
     V=(V_1,V_2),\quad W=(W_1,W_2).
    $$
For fixed $V\in\mathcal K$, we associate the linear operator $S_V:\K \to \K$ as
    $$S_VW :=\left(\sum_{i=1}^2V_i^*W_i,\; \sum_{i=1}^2W_iV_i^*\right). $$
In the case $V_1,V_2$ are \emph{Hilbert--Schmidt orthonormal}, i.e., $\Tr(V_i^*V_j)=\delta_{ij}$, we have $\|V\|_{\K}^2=2$ and hence the operators
    $$P_V:=\frac12|V\rangle\langle V|, \qquad P_{V^{\perp}}:=I_{\K}-P_V$$
define orthogonal projections onto $\Comp V$ and its orthogonal complement $V^\perp=\{W\in\mathcal K:\langle V,W\rangle_{\mathcal K}=0\},$ respectively. Let us further introduce three intermediate quantities
\begin{equation} \label{eq:hgK}
\begin{aligned}
    h_V&:= \frac12\langle V,S_V^*S_VV\rangle \in \Comp,\\
    |g_V\ra&:=\frac1{\sqrt2}P_{V^\perp}S_V^*S_V\ket{V} \in V^{\perp},\\
    K_V&:=\left.P_{V^\perp}S_V^*S_VP_{V^\perp}\right|_{V^\perp} \in B(V^{\perp}).
\end{aligned}
\end{equation}

\begin{lemma}
\label{lem:endpoint-reduction}
For $d\geq2$, the following statements are equivalent.

\begin{enumerate}
\item The Werner state $\rho_{-1/2}$ is 2-copy undistillable.

\item For every Hilbert--Schmidt orthonormal pair
$V_1,V_2\in M_d(\mathbb C)$ and every
$W=(W_1,W_2)\in\mathcal K$,
\begin{equation} 
 \left\|\sum_{i=1}^2V_i^*W_i\right\|_2^2
 +\left\|\sum_{i=1}^2W_iV_i^*\right\|_2^2
 \leq
 2\sum_{i=1}^2\|W_i\|_2^2
 +\frac12\left|\sum_{i=1}^2\operatorname{Tr}(V_i^*W_i)\right|^2.
 \label{eq:matrix-endpoint}
\end{equation}

\item For every Hilbert--Schmidt orthonormal pair
$V=(V_1,V_2)$,
\begin{equation}
S_V^*S_V\leq 2I_{\mathcal K}+P_V.
 \label{eq:operator-endpoint}
\end{equation}

\item For every Hilbert--Schmidt orthonormal pair
$V=(V_1,V_2)$, the quantities in \eqref{eq:hgK} satisfy
\begin{equation}
 h_V\leq3,
 \qquad
 2I_{V^\perp}-K_V\geq0,
 \qquad
 |g_V\rangle\langle g_V|
 \leq(3-h_V)(2I_{V^\perp}-K_V).
 \label{eq:schur-conditions}
\end{equation}
\end{enumerate}
\end{lemma}

\begin{proof}
Let us take a vector $|\psi\rangle = |v_1\ra_{12}|w_1\ra_{34}+|v_2\ra_{12}|w_2\ra_{34}$ as in Eq. \eqref{eq-bipartite-Schmidt}, and we further associate
    $$|v_i\rangle=|\vect V_i\rangle=(V_i\otimes I_d)|\Omega_d\rangle,  \qquad |w_i\rangle=\overline{|\vect W_i\rangle}=(I_d\otimes W_i^*)|\Omega_d\rangle, \qquad i=1,2.$$
The entrywise conjugation in the second formula is a convenient
parametrization of the arbitrary vectors $w_i$ to avoid transposes in
the matrix inequality below. Note that the orthonormality of $v_1,v_2$ is
precisely
$\operatorname{Tr}(V_i^*V_j)=\delta_{ij}$. Let us also write $P_\Omega:=|\Omega_d\rangle\langle\Omega_d|$ for simplicity. Then the vectorization convention gives the contraction identities
\begin{align*}
    (I_{13}\otimes\langle\Omega_d|_{24})
    \bigl(|\vect X\rangle_{12}
       \otimes\overline{|\vect Y\rangle}_{34}\bigr)
    &=|\vect (XY^*)\rangle_{13},\\
    (\langle\Omega_d|_{13}\otimes I_{24})
    \bigl(|\vect X\rangle_{12}
       \otimes\overline{|\vect Y\rangle}_{34}\bigr)
    &=\overline{|\vect (X^*Y)\rangle}_{24}.
\end{align*}
Consequently, one can write
\begin{align*}
 \|\psi\|^2
 &=\sum_{i=1}^2\|W_i\|_2^2, & \langle\psi|P_\Omega^{(13)}\otimes P_\Omega^{(24)}|\psi\rangle
 &=\left|\sum_{i=1}^2\operatorname{Tr}(V_i^*W_i)\right|^2, \\
 \langle\psi|P_\Omega^{(13)}\otimes I_{24}|\psi\rangle &=\left\|\sum_{i=1}^2V_i^*W_i\right\|_2^2, &
 \langle\psi|I_{13}\otimes P_\Omega^{(24)}|\psi\rangle &=\left\|\sum_{i=1}^2W_iV_i^*\right\|_2^2. 
\end{align*}
Since $\rho_{-1/2}^\Gamma
 =\frac{1}{d^2-d/2}(I_{d^2}-\frac12P_\Omega),$
we have
$$
 (\rho_{-1/2}^\Gamma)^{\otimes2}
 =\frac{1}{2(d^2-d/2)^2}
 \left[
 2I-\bigl(P_\Omega^{(13)}\otimes I_{24}
                  +I_{13}\otimes P_\Omega^{(24)}\bigr)
 +\frac12 P_\Omega^{(13)}\otimes P_\Omega^{(24)}
 \right].
$$
Substitution shows that
$\langle\psi|(\rho_{-1/2}^\Gamma)^{\otimes2}|\psi\rangle\geq0$
is equivalent to \eqref{eq:matrix-endpoint}, which shows the equivalence (1) $\Leftrightarrow$ (2).

Furthermore, by definition,
$$
 \langle W,S_V^*S_VW\rangle
 =\|S_VW\|_{\mathcal K}^2
 =\left\|\sum_iV_i^*W_i\right\|_2^2
  +\left\|\sum_iW_iV_i^*\right\|_2^2,
$$
whereas
$$
 \langle W,(2I_{\mathcal K}+P_V)W\rangle
 =2\sum_i\|W_i\|_2^2
  +\frac12\left|\sum_i\operatorname{Tr}(V_i^*W_i)\right|^2.
$$
Thus (2) and (3) are equivalent.

Finally, relative to the orthogonal decomposition
$\mathcal K=\mathbb CV\oplus V^\perp$, the operator in
\eqref{eq:operator-endpoint} has block form
\begin{equation}
 2I_{\mathcal K}+P_V-S_V^*S_V
 =\begin{pmatrix}
 3-h_V&-\langle g_V|\\
 -|g_V\rangle&2I_{V^\perp}-K_V
 \end{pmatrix}.
 \label{eq:block-form}
\end{equation}
If $3-h_V>0$, the Schur-complement criterion says that
\eqref{eq:block-form} is positive if and only if
$$
 2I_{V^\perp}-K_V
 -\frac1{3-h_V}|g_V\rangle\langle g_V|\geq0.
$$
If $3-h_V=0$, positivity forces $g_V=0$, and the remaining condition
is $2I_{V^\perp}-K_V\geq0$.  If $3-h_V<0$, positivity is impossible. That proves the equivalence of (3) and (4).
\end{proof}

We remark that the equivalence (1) $\Leftrightarrow$ (2) is already noted in \cite{Djok16}. Furthermore, this equivalence can also be derived from another equivalent \textit{partial trace inequality}
\begin{equation} \label{eq:PTIneq}
    \|C_{A}\|_2^2+\|C_B\|_2^2\leq 2\|C_{AB}\|_2^2 + \frac{1}{2}|\Tr(C_{AB})|^2
\end{equation}
for any bipartite matrix $C_{AB}\in M_d(\Comp)^{\otimes 2}$ of rank at most $2$, where $C_A:=(\id\otimes \Tr)(C_{AB})$ and $C_B:=(\Tr\otimes \id)(C_{AB})$. Indeed, two inequalities \eqref{eq:matrix-endpoint} and \eqref{eq:PTIneq} are connected straightforwardly from
    $$|\vect C_{AB}\ra_{A_1A_2:B_1B_2} = |\vect V_1\ra_{A_1:A_2}\overline{|\vect W_1\ra}_{B_1:B_2} + |\vect V_2\ra_{A_1:A_2}\overline{|\vect W_2\ra}_{B_1:B_2}$$
Similar type of inequalities as in \eqref{eq:PTIneq} have been extensively studied in \cite{Ric25,RW25}, where it was crucial in proving the 2-copy undistillability of Werner states in the range $\alpha\geq -\frac{1}{3}$. Our approach takes further steps to reduce it to the operator inequalities \eqref{eq:operator-endpoint} and \eqref{eq:schur-conditions}, which are fully proved in the following section.

\subsection{The antisymmetric projection on a symmetric subspace}

The remainder of the section is devoted to proving the inequalities \eqref{eq:schur-conditions}. Specifically, we aim to provide the explicit connection between the intermediate variables $h_V,g_V$ and $K_V$ and the Hilbert--Schmidt norm of the restricted projection
    $$\Pi_{\A}^{(13)}\otimes \Pi_{\A}^{(24)}\big|_{\Hi\vee \Hi} : \Hi^{(12)}\vee \Hi^{(34)} \to (\Comp^d)^{\otimes 4}$$
of the tensor antisymmetric projection $Q:=\Pi_{\A}^{(13)}\otimes\Pi_{\A}^{(24)}$ considered in Section \ref{sec-symm-antisymm}, where $\Hi={\rm span}\{|\vect V_1\ra, |\vect V_2\ra\}$. 
Equivalently, for $V=(V_1,V_2)\in \K$, we associate an operator $U_V:\Comp^2\to (\Comp^d)^{\otimes 2}$ by linearly extending the correspondence
    $$U_V|i\ra:=|\vect V_i\ra, \quad i=1,2.$$
Note that $U_V$ is well-defined for \emph{arbitrary} $V\in \K$ with the property
    $$ \Tr(U_V^*U_W )=\langle V, W\rangle_{\mathcal{K}}, \quad V,W\in \K.$$
Furthermore, $U_V$ is an isometry whenever $(V_1,V_2)$ is an orthonormal pair, in which case the operator $U_V\otimes U_V$ further induces an isometry from $\Comp^2\vee \Comp^2$ onto $\Hi^{(12)}\vee \Hi^{(34)}$. This motivates the linear operator 
    $$\Fop_V:=
    Q(U_V\otimes U_V)\big|_{\Comp^2\vee\Comp^2}:\Comp^2\vee\Comp^2\to (\Comp^d)^{\otimes4}, \qquad V\in \K.$$
    
We first derive the general Hilbert--Schmidt inner product formula and provide the explicit connection between $\Fop_V$ and the quantity $h_V$.

\begin{lemma}[Projection kernel and zeroth-order identity]
\label{lem:projection-kernel}
For any $V,V'\in\mathcal K$,
    \begin{align}\langle\Fop_V,\Fop_{V'}\rangle_{\mathrm{HS}}
 =\frac14\Big(
 \bigl(\operatorname{Tr}(U_V^*U_{V'})\bigr)^2
 +\operatorname{Tr}\bigl((U_V^*U_{V'})^2\bigr)
 -\langle S_{V'}V,S_VV'\rangle_{\mathcal K}
 \Big).\label{eq:zero}\end{align}
 In particular, if $V_1,V_2$ are Hilbert--Schmidt orthonormal, then
 \begin{equation}
  \|\Fop_V\|_{\mathrm{HS}}^2
  =\frac{3-h_V}{2}.
  \label{eq:zeroth-variation}
\end{equation}
\end{lemma}

\begin{proof}
Let $\ket{\xi_{11}}=|11\rangle,$ $\ket{ \xi_{12}}=\frac{|12\rangle+|21\rangle}{\sqrt2}$ and $\ket{\xi_{22}}=|22\rangle$,
which form an orthonormal basis of
$\Comp^2\vee \Comp^2$. Then all the vectors
    $$\ket{\psi_{ij}}:=(U_V\otimes U_V)\ket{\xi_{ij}},
 \qquad
 \ket{\phi_{ij}}:=(U_{V'}\otimes U_{V'})\ket{\xi_{ij}}, \qquad 1\leq i\leq j\leq 2,$$
are symmetric under $F^{(13)}F^{(24)}$. Therefore, the relation \eqref{eq-symm-identity} can be applied to have
\begin{align*}
 \langle\Fop_V,\Fop_{V'}\rangle_{\mathrm{HS}}=&\sum_{1\le i\le j\le 2}\langle  \psi_{ij}| Q| \phi_{ij}\rangle\\
=&\frac14\sum_{1\le i\le j\le 2}\langle  \psi_{ij}| I-F^{(13)}-F^{(24)}+F^{(13)}F^{(24)}| \phi_{ij}\rangle \\
 =&\frac12\sum_{1\le i\le j\le 2}\langle  \psi_{ij}| I-F^{(24)}| \phi_{ij}\rangle= \frac12T_0-\frac12T_1,
\end{align*}
where $\displaystyle T_0:=\sum_{1\leq i\leq j\leq2}\langle\psi_{ij}|\phi_{ij}\rangle$ and $\displaystyle  T_1:=\sum_{1\leq i\leq j\leq2}\langle\psi_{ij}|F^{(24)}|\phi_{ij}\rangle.$
First of all, since the orthogonal projection
onto $\Comp^2\vee \Comp^2$ is $\Pi_{\mathcal{S}}^{(2)} = \sum_{1\leq i\leq j\leq 2}|\xi_{ij}\ra\la \xi_{ij}| = (I+F_2)/2$, one can simplify
\begin{align*}
    T_0=&\sum_{1\leq i\leq j\leq2}\operatorname{Tr}(\ket{\xi_{ij}}\bra{\xi_{ij}} (U_V^*U_{V'}\otimes U_V^*U_{V'}))\\
    =&\frac{1}{2}\operatorname{Tr}( (I+F_2)U_V^*U_{V'}\otimes U_V^*U_{V'})\\
    =&\frac12\left((\operatorname{Tr}U_V^*U_{V'})^2+\operatorname{Tr}((U_V^*U_{V'})^2)\right).
\end{align*}
Furthermore, for the flip contribution, the direct computation using Lemma \ref{lemma:4trace} implies that
\begin{align*}
    T_1 &= \Tr(V_1^*V_1'V_1^*V_1') + \Tr(V_2^*V_2'V_2^*V_2') \\
    & \quad +
     \frac{1}{2} \Tr(V_1^*V_1'V_2^*V_2'+ V_1^*V_2'V_2^*V_1' + V_2^*V_1'V_1^*V_2' + V_2^*V_2'V_1^*V_1')\\
    &= \frac{1}{2}\sum_{i,j=1}^2 \Tr(V_i^*V_i'V_j^*V_j') + \frac{1}{2} \sum_{i,j=1}^2 \Tr(V_i'V_i^*V_j'V_j^*) \\
    &=\frac{1}{2} \la S_{V'}V, S_V V'\ra_{\K}.
\end{align*}
Combining all the above shows the assertion \eqref{eq:zero}. The second assertion \eqref{eq:zeroth-variation} for $V=V'$ follows from that
$U_V$ is an isometry (i.e. $U_V^*U_V = I_2$) if $(V_1,V_2)$ are orthonormal.
\end{proof}

\subsection{First and second variations}

Because $V\mapsto\Fop_V$ is quadratic, for every real $t$, we have
\begin{align*}
 \Fop_{V+tW}
 =&\Fop_V+t\,D\Fop_V[W]+t^2\Fop_W,
 \qquad\\
 D\Fop_V[W]
 :=&\left.Q(U_W\otimes U_V+U_V\otimes U_W)
   \right|_{\Comp^2\vee \Comp^2}.
\end{align*}
Equivalently, $D\Fop_V[W]$ can be naturally defined via the \emph{directional derivative} 
\begin{equation} \label{eq-DQV}
    D\Fop_V[W] = \frac{d}{dt}\Big|_{t=0}\Fop_{V+tW} = \left.Q(U_W\otimes U_V+U_V\otimes U_W)
   \right|_{\Comp^2\vee \Comp^2}.
\end{equation}
The above first-order variation allows us to identify the off-diagonal term in the Schur
complement.

\begin{lemma} [First variation] \label{lem:first-variation}
For any three pairs $V, V',W\in \K$, we have
\begin{equation} \label{eq:first-variation}
    \la \Fop_V,D\Fop_{V'}[W]\ra_{\rm HS} = \frac{1}{2}\big( \Tr(U_V^*U_{V'}) \Tr(U_V^*U_{W}) + \Tr(U_V^*U_{V'}U_V^*U_{W}) - \la S_{V'}V,S_V W\ra_{\K}  \big).
\end{equation}
 In particular, for $V=(V_1,V_2)$ a Hilbert--Schmidt orthonormal pair and for
 $W\in V^\perp$, we have
\begin{equation}
  \langle\Fop_V,D\Fop_V[W]\rangle_{\mathrm{HS}}
  =-\frac1{\sqrt2}\langle g_V,W\rangle_{\mathcal K}.
  \label{eq:first-variation-pairing}
 \end{equation}
\end{lemma}
\begin{proof}
By the definition of $D\Fop_V$, we can proceed using Lemma \ref{lem:projection-kernel} as follows:
\begin{align*}
    &\la \Fop_V,D\Fop_{V'}[W]\ra_{\rm HS} = \frac{d}{dt}\Big|_{t=0} \la \Fop_{V}, \Fop_{V'+tW}\ra_{\rm HS}\\
    &= \frac{1}{4} \frac{d}{dt}\Big|_{t=0}\Big( \big(\Tr(U_V^*U_{V'})+ t\Tr(U_V^*U_{W})\big)^2 + \Tr((U_V^*U_{V'}+ tU_V^*U_W)^2)\\
    &\qquad\qquad\qquad- \la S_{V'}V+ tS_{W}V, S_V V' + tS_V W\ra_{\K}   \Big)\\
    &=\frac{1}{4} \big( 2\Tr(U^*_VU_{V'})\Tr(U_V^*U_W) + 2 \Tr(U_V^*U_{V'}U_V^*U_W) - \la S_{V'}V,S_V W\ra_{\K} - \la S_{W}V,S_V V'\ra_{\K}\big)\\
    &= \frac{1}{2}\left( \Tr(U_V^*U_{V'}) \Tr(U_V^*U_{W}) + \Tr(U_V^*U_{V'}U_V^*U_{W}) - \la S_{V'}V,S_V W\ra_{\K}  \right).
\end{align*}
In the last equality above, we used the relation 
    $$\la S_{V'}V, S_V W\ra_{\K} = \la S_{W}V, S_V V'\ra_{\K} = \sum_{i,j=1}^2 \big(\Tr(V_i^*V_i'V_j^*W_j) + \Tr(V_i'V_i^* W_j V_j^*)\big).$$
The second assertion  follows from \eqref{eq:first-variation} with the relations $U_V^* U_V = I_2$ and $\Tr(U_V^*U_W) = 0$.
\end{proof}

Finally, the second-order variation controls the transverse block $K_V = P_{V^{\perp}}S_V^* S_V P_{V^\perp}$.

\begin{lemma} [Second variation] \label{cor-DQV-K}
For any orthonormal pair $V=(V_1,V_2)\in \K$ and for any $W\in V^{\perp}$, we have
\begin{align}
 \|D\Fop_V[W]\|_{\mathrm{HS}}^2
 &=\frac12\Bigl(
 3\|W\|_{\mathcal K}^2
 +\|U_V^*U_W\|_2^2
 -\langle W,  S_V^*S_VW\rangle_{\mathcal K}
 -\langle S_VV,S_WW\rangle_{\mathcal K}
 \Bigr)
 \label{eq:differential-norm}\\
 &\leq\langle W,(2I_{V^\perp}-K_V)W\rangle_{\mathcal K} = 2\|W\|_{\K}^2 - \la W, S_V^* S_V W\ra_{\K}.
 \label{eq:second-variation-bound}
\end{align}
\end{lemma}
\begin{proof}
By repeating the proof argument of Lemma \ref{lem:first-variation}, together using Eq. \eqref{eq:first-variation}, we obtain
\begin{align*}
    &\|D\Fop_V[W]\|_{\rm HS}^2 = \frac{d}{dt}\Big|_{t=0}\la \Fop_{V+tW},D\Fop_{V}[W]\ra_{\rm HS} \\ 
    &= \frac{1}{2}\big(\Tr(U_V^*U_V) \Tr(U_W^*U_W) + \Tr(U_W^*U_V)\Tr(U_V^*U_W) \\
    & \qquad + \Tr(U_V^*U_V U_W^*U_W + U_W^*U_V U_V^* U_W) -\la S_V W , S_V W\ra_{\K} - \la S_V V , S_W W\ra_{\K}\big)\\
    &= \frac{1}{2}(3\|W\|_{\K}^2 + \|U_V^*U_W\|_2^2 - \la W,S_V^* S_V W\ra_{\K} - \la S_V V,S_W W\ra_{\K}),
\end{align*}
where we used the relations $U_V^*U_V = I_2$, $\Tr(U_V^*U_W)=0$, and $\Tr(U_W^*U_W)=\|W\|^2$ in the last equality. This shows Eq. \eqref{eq:differential-norm}.

It remains to prove \eqref{eq:second-variation-bound} which is equivalent to
\begin{equation}
 \|W\|_{\mathcal K}^2
 -\|U_V^*U_W\|_2^2
 -\|S_VW\|_{\mathcal K}^2
 +\langle S_VV,S_WW\rangle_{\mathcal K}
 \geq0.
 \label{eq:variation-remainder}
\end{equation}
The vector below is chosen so that its norm and its expectation under
the full copy flip reproduce the first two terms on the left-hand side
of \eqref{eq:variation-remainder}.  Let us take
    $$|\psi\ra = |\psi_{V,W}\ra
    :=|\vect V_1\ra_{12}|\vect W_2\ra_{34}
      -|\vect V_2\ra_{12}|\vect W_1\ra_{34}.$$
Obviously we have ${\rm SR}_{12:34}(\psi)\leq 2$. Since $V_1,V_2$ are orthonormal, we first have 
    $$\|\psi\|^2=\|W_1\|^2+\|W_2\|^2 = \|W\|_{\K}^2.$$
Furthermore, the action of the full flip operator $F^{(12:34)}=F^{(13)}F^{(24)}$ gives rise to 
\begin{align*}
    \la\psi|F^{(13)}F^{(24)}|\psi\ra &= |\Tr(V_1^*W_2)|^2 -\Tr(V_1^*W_1) \Tr(W_2^*V_2) - \Tr(V_2^*W_2)\Tr(W_1^*V_1) + |\Tr(V_2^*W_1)|^2\\
    &= |\Tr(V_1^*W_2)|^2 + |\Tr(V_1^*W_1)|^2 + |\Tr(V_2^*W_2)|^2 + |\Tr(V_2^*W_1)|^2\\
    &= \|U_V^*U_W\|_2^2,
\end{align*}
where we applied the orthogonality condition $\Tr(W_1^*V_1)=-\Tr(W_2^*V_2)$ in the second equality. On the other hand, Lemma \ref{lemma:4trace} implies that
\begin{align*}
    \la\psi|F^{(13)}|\psi\ra &=\Tr(V_1^*W_2W_2^*V_1) + \Tr(V_2^*W_1W_1^*V_2) -\Tr(V_1^*W_1W_2^*V_2) -\Tr(V_2^*W_2W_1^*V_1),\\
    \la\psi|F^{(24)}|\psi\ra 
    &=\Tr(V_1^*V_1W_2^*W_2) + \Tr(V_2^*V_2W_1^*W_1) - \Tr(V_1^*V_2W_2^*W_1) - \Tr(V_2^*V_1W_1^*W_2),
\end{align*}
and therefore, 
\begin{align*}
    \la \psi|F^{(13)}+F^{(24)}|\psi\ra &= \sum_{i,j=1}^2\Tr(V_i^*V_iW_j^*W_j + V_iV_i^*W_jW_j^* - V_i^*W_i W_j^*V_j - W_iV_i^*V_jW_j^*)\\
    &= \la S_V V, S_W W\ra_{\K} - \la S_V W , S_V W\ra_{\K}.
\end{align*}
Consequently, Eq. \eqref{eq:variation-remainder} is further equivalent to 
    $$\la \psi|I+F^{(13)}+F^{(24)}-F^{(13)}F^{(24)}|\psi\ra\geq 0.$$
Since $I+F^{(13)}+F^{(24)}-F^{(13)}F^{(24)} = 2(I-2Q) = 2(I-2\Pi_{\A}^{(13)}\otimes \Pi_{\A}^{(24)})$, the application of Theorem \ref{thm-AntisymmProj-S2norm} concludes the proof of Eq. \eqref{eq:second-variation-bound}.
\end{proof}

We are now ready to prove our main Theorem \ref{thm-main}.
\begin{proof}[\textbf{Proof of Theorem \ref{thm-main}}]
It suffices to show the equivalent operator conditions \eqref{eq:schur-conditions} from Lemma \ref{lem:endpoint-reduction}. First of all, the relations \eqref{eq:zeroth-variation} and \eqref{eq:first-variation} in Lemmas \ref{lem:projection-kernel} and \ref{lem:first-variation} guarantee that $h_V\leq 3$ and $2I_{V^{\perp}}-K_V\geq 0$. Now take $W\in V^\perp$. Combining 
 \eqref{eq:zeroth-variation}, \eqref{eq:first-variation}, and Lemma \ref{cor-DQV-K} with the Cauchy--Schwarz inequality gives
\begin{align*}
    |\la g_V,W\ra_{\K}|^2 &=2|\la\Fop_V,D\Fop_V[W]\ra_{\rm HS}|^2 \\
    &\leq 2\|\Fop_V\|_{\rm HS}^2 \|D\Fop_V[W]\|_{\rm HS}^2 \\&\leq (3-h_V)\la W,(2I_{V^\perp}-K_V)W\ra_{\K}.
\end{align*}
This shows the operator inequality $|g_V\ra\la g_V|\leq (3-h_V)(2I_{V^{\perp}}-K_V)$. Consequently, Lemma \ref{lem:endpoint-reduction} (4) completes the proof.
\end{proof}

\section{Summary and outlook}
\label{sec:conclusion}

We have obtained a complete analytic characterization of the 2-copy
distillability of Werner states in every local dimension. More precisely,
Theorem~\ref{thm-main} shows that the 2-copy distillability threshold
coincides exactly with the 1-copy threshold:
    $$\rho_\alpha \text{ is 2-copy distillable} \quad\Longleftrightarrow\quad \alpha<-\frac12.$$
Consequently, throughout the NPT interval $-\frac12\leq\alpha<-\frac1d,$ a second copy does not activate distillability. This settles the first
nontrivial finite-copy instance of the Werner-state
distillability problem uniformly in the dimension and removes the gap
between the previously known sufficient conditions \cite{DCLB00,RW25} and the conjectured
sharp threshold \cite{DCLB00,PPHH10,HRZ22}, even beyond the dimension $d\geq 4$. Since Werner states form the crucial test family for
the NPT distillation problem \cite{HH99}, the result provides a new exact benchmark
for understanding how negative partial transpose, Schmidt rank, and the
number of available copies interact.

As a comparison between our solution with previous approaches, we briefly mention the special case $d=4$ which has been extensively studied \cite{PPHH10,JK10,LC26}. In this case, up to a factor, the tensor power of the partial-transposed endpoint Werner state $(\rho_{-1/2}^{(4)})^{\Gamma}$ becomes
    $$(\rho_{-1/2}^{\Gamma})^{\otimes 2} \propto I_{4^4} - 2P_- = I-2\big(\om_4^{(13)}\otimes (I_{4^2}-\om_4)^{(24)} + (I_{4^2}-\om_4)^{(13)}\otimes \om_4^{(24)}\big),$$
where $\om_4:= \frac{1}{4}|\Om_4\ra\la \Om_4|$ denotes the maximally entangled pure state. Therefore, $\rho_{-1/2}^{(4)}$ is $2$-copy undistillable if and only if 
    $$\|P_-\|_{S(2),12:34} =  \sup_{\substack{\|\psi\|=1\\ {\rm SR}_{12:34}(\psi)\leq 2}}\la \psi|P_-|\psi\ra\leq \frac{1}{2}.$$
In particular, our Theorem \ref{thm-main} resolves this optimization and many other equivalent formulations such as
\begin{center}
    $\sigma_1(X)^2+\sigma_2(X)^2\leq \frac{1}{2}$, \qquad $X=A\otimes I_4+I_4\otimes B$, \quad whenever \\
    $A,B\in M_4(\Comp)$, \quad    $\Tr(A)=\Tr(B)= 0$, \quad  and \quad $\|A\|_2^2+\|B\|_2^2 = \frac{1}{4}$,
\end{center}
as considered in \cite{PPHH10,QCCS21}. Instead of showing directly the above optimization problems, our key
ingredient is the sharp, dimension-independent estimate involving the tensor antisymmetric projection
$$
 \left\|\Pi_{\mathcal A}^{(13)}\otimes
              \Pi_{\mathcal A}^{(24)}\right\|_{S(2),12:34}
 =\frac12,
$$
which converts a Schmidt-rank constraint into a geometric statement about
symmetric and antisymmetric tensor subspaces. The boundary case $\alpha=-1/2$ is
then recast as a block-operator inequality, whose diagonal,
off-diagonal, and transverse components are identified with the value,
first variation, and second variation of the restricted antisymmetric
projection. We expect this
perspective to be useful in other problems involving $S(k)$-norms \cite{JK10,JKPP11},
Schmidt-number witnesses \cite{TH00,SLB01}, and tensor powers of positive maps \cite{MHRW16}.

It is important
to stress that 2-copy undistillability is not yet NPT bound
entanglement. 
The immediate further question is the $r$-copy problem for general $r\ge 3$. At the
endpoint, this asks whether
    $$\langle\psi| \bigl(\rho_{-1/2}^{\Gamma}\bigr)^{\otimes r}  |\psi\rangle\geq 0  \qquad  \text{for every }\operatorname{SR}(\psi)\leq 2  $$
for all $r\geq3$. A positive answer for every $r$ would imply NPT bound
entanglement within the Werner family. This problem still remains open.

\bigskip
{
\noindent\textbf{Note added.}
After the first version of this work appeared on arXiv, three contemporaneous preprints 
obtained the same two-copy threshold using complementary partial-trace and
matrix-inequality methods \cite{SC26,FHPV26,BGH26}.}

\bigskip

\noindent{\bf AI Statement.} The authors acknowledge the use of AI tools, including ChatGPT, for language polishing, LaTeX editing, and exploratory mathematical discussions during the development and preparation of this manuscript. Some ideas used in the proof-development process arose during interactions with GPT-5.5 and GPT-5.6 Sol. Suggestions from these AI interactions were subsequently examined, reformulated, incorporated into the manuscript, and independently verified by the authors. The authors take full responsibility for all mathematical content in the final manuscript.

\bigskip

\noindent{\bf Acknowledgments.} We thank Pablo Costa Rico for helpful comments on the first version of this
manuscript. J. Fu, L. Gao, and S.-J. Park  are partially supported by the National Natural Science Foundation of China (grant No.~12401163) and the Department of Science and Technology of Hubei Province (Project No.2025EHA041, Project No.2025AFA044). S.-J. Park acknowledges the Hongyi Postdoc Fellowship from Wuhan University.

\bibliography{references}
\bibliographystyle{alpha}
\bigskip
\hrule
\bigskip

\end{document}